\documentclass[a4paper]{amsart}\usepackage[a4paper,margin=2.5cm,includefoot]{geometry}
\usepackage[utf8]{inputenc}
\usepackage[english]{babel}
\usepackage[T1]{fontenc}

\usepackage{amsmath,amssymb,mathtools,bm}
\usepackage{csquotes}
\usepackage{amstext}
\usepackage{amsthm}
\usepackage{bbm}
\usepackage{enumerate}
\usepackage{hyperref}
\usepackage[nameinlink,capitalise]{cleveref}
\usepackage{braket}
\usepackage{dsfont}
\usepackage{color}
\usepackage{tikz}
\usepackage{pgfplots}

\makeatletter
\def\@seccntformat#1{%
	\protect\textup{\protect\@secnumfont
		\ifnum\pdfstrcmp{subsection}{#1}=0 \bfseries\fi
		\ifnum\pdfstrcmp{subsubsection}{#1}=0 \itshape\fi
		\csname the#1\endcsname
		\protect\@secnumpunct
	}%
}
\renewcommand{\@upn}{}
\DeclareRobustCommand{\crefnosort}[1]{%
	\begingroup\@cref@sortfalse\cref{#1}\endgroup
}
\makeatother

\numberwithin{equation}{section}
\newtheorem{thm}{Theorem}[section]
\newtheorem{lem}[thm]{Lemma}
\newtheorem{prop}[thm]{Proposition}
\newtheorem{cor}[thm]{Corollary}

\theoremstyle{definition}

\renewcommand*{\thehyp}{\Alph{hyp}}

\theoremstyle{remark}
\newtheorem{rem}[thm]{Remark}

\crefname{hyp}{Hypothesis}{Hypotheses}
\Crefname{hyp}{Hypothesis}{Hypotheses}
\crefname{lem}{Lemma}{Lemmas}
\Crefname{lem}{Lemma}{Lemmas}
\crefname{thm}{Theorem}{Theorems}
\Crefname{thm}{Theorem}{Theorems}
\crefname{prop}{Proposition}{Propositions}
\Crefname{prop}{Proposition}{Propositions}
\crefname{enumi}{}{}
\Crefname{enumi}{}{}
\creflabelformat{enumi}{#2(#1)#3}
\crefname{equation}{}{}
\Crefname{equation}{}{}
\crefname{rem}{Remark}{Remarks}
\Crefname{rem}{Remark}{Remarks}

\makeatletter
\renewcommand{\@upn}{} 
\makeatother

\usepackage{etoolbox}
\makeatletter
\patchcmd{\endthm}{\@endpefalse}{}{}{}
\patchcmd{\endcor}{\@endpefalse}{}{}{}
\patchcmd{\endlem}{\@endpefalse}{}{}{}
\patchcmd{\endprop}{\@endpefalse}{}{}{}
\patchcmd{\endproof}{\@endpefalse}{}{}{}
\makeatother

\usepackage[inline]{enumitem}
\newlist{enumthm}{enumerate}{1} 
\setlist[enumthm]{label=\upshape(\roman*),ref=\thethm~(\roman*)}  
\crefalias{enumthmi}{thm} 
\newlist{enumcor}{enumerate}{1}
\setlist[enumcor]{label=\upshape(\roman*),ref=\thecor~(\roman*)}
\crefalias{enumcori}{cor}
\newlist{enumlem}{enumerate}{1}
\setlist[enumlem]{label=\upshape(\roman*),ref=\thelem~(\roman*)}
\crefalias{enumlemi}{lem}
\newlist{enumprop}{enumerate}{1}
\setlist[enumprop]{label=\upshape(\roman*),ref=\theprop~(\roman*)}
\crefalias{enumpropi}{prop}
\newlist{enumhyp}{enumerate}{1}
\setlist[enumhyp]{label=\upshape(\roman*),ref=\thehyp~(\roman*)}
\crefalias{enumhypi}{hyp}
\newlist{enumproof}{enumerate*}{1}
\setlist[enumproof]{label=\upshape(\roman*)}
\newlist{enumdef}{enumerate}{1}
\setlist[enumdef]{label=\upshape(\roman*),ref=\thedefn~(\roman*)}
\crefalias{enumdefi}{defn}

\makeatletter
\newcounter{subcreftmpcnt} %
\newcommand\romansubformat[1]{(\roman{#1})} 
\def\subcref{\@ifstar\@@subcref\@subcref}
\newcommand\@subcref[2][\romansubformat]{%
	\ifcsname r@#2@cref\endcsname
	\cref@getcounter {#2}{\mylabel}%
	\setcounter{subcreftmpcnt}{\mylabel}%
	\hyperref[#2]{\romansubformat{subcreftmpcnt}}%
	\else ?? \fi}   
\newcommand\@@subcref[2][\romansubformat]{%
	\ifcsname r@#2@cref\endcsname
	\cref@getcounter {#2}{\mylabel}%
	\setcounter{subcreftmpcnt}{\mylabel}%
	\romansubformat{subcreftmpcnt}%
	\else ?? \fi}   
\makeatother

\makeatletter
\DeclareRobustCommand{\crefnosort}[1]{%
	\begingroup\@cref@sortfalse\cref{#1}\endgroup
}
\makeatother

\def\endstepsymbol{$\lozenge$}
\def\endclaimsymbol{$\lozenge$}
\newcounter{proofstep}
\AtBeginEnvironment{proof}{\setcounter{proofstep}{0}}

\crefname{proofstep}{Step}{Steps}
\Crefname{proofstep}{Step}{Steps}
\newcounter{proofclaim}
\AtBeginEnvironment{proof}{\setcounter{proofclaim}{0}}

\crefname{proofclaim}{Claim}{Claims}
\Crefname{proofclaim}{Claim}{Claims}

\newcommand{\cC}{{\mathcal C}}
\newcommand{\cD}{{\mathcal D}}\newcommand{\cF}{{\mathcal F}}

\newcommand{\BN}{{\mathbb N}}
\newcommand{\BR}{{\mathbb R}}

\usepackage{dsfont} 

\usepackage{mathrsfs}

\newcommand{\sG}{{\mathscr G}}

\newcommand{\sfP}{{\mathsf P}}

\newcommand{\sfc}{{\mathsf c}}
\newcommand{\sfd}{{\mathsf d}}\newcommand{\sfe}{{\mathsf e}}
\newcommand{\sfg}{{\mathsf g}}

\newcommand{\sfs}{{\mathsf s}}

\newcommand{\IN}{\BN}\newcommand{\IR}{\BR}
\newcommand{\R}{\BR}

\newcommand{\eps}{\varepsilon}\newcommand{\ph}{\varphi}

 \renewcommand{\d}{\sfd}

\DeclareMathOperator*{\essinf}{ess\,inf}

\newcommand{\wh}[1]{\widehat{#1}}\newcommand{\wt}[1]{\widetilde{#1}}

\DeclareFontFamily{U}{mathx}{\hyphenchar\font45}
\DeclareFontShape{U}{mathx}{m}{n}{
	<5> <6> <7> <8> <9> <10>
	<10.95> <12> <14.4> <17.28> <20.74> <24.88>
	mathx10
}{}
\DeclareSymbolFont{mathx}{U}{mathx}{m}{n}
\DeclareFontSubstitution{U}{mathx}{m}{n}
\DeclareMathAccent{\widecheck}{0}{mathx}{"71}
\DeclareMathAccent{\wideparen}{0}{mathx}{"75}

\DeclareFontFamily{OMX}{MnSymbolE}{}
\DeclareFontShape{OMX}{MnSymbolE}{m}{n}{
	<-6>  MnSymbolE5
	<6-7>  MnSymbolE6
	<7-8>  MnSymbolE7
	<8-9>  MnSymbolE8
	<9-10> MnSymbolE9
	<10-12> MnSymbolE10
	<12->   MnSymbolE12}{}
\DeclareSymbolFont{mnlargesymbols}{OMX}{MnSymbolE}{m}{n}
\SetSymbolFont{mnlargesymbols}{bold}{OMX}{MnSymbolE}{b}{n}
\DeclareMathDelimiter{\llangle}{\mathopen}{mnlargesymbols}{'164}{mnlargesymbols}{'164}
\DeclareMathDelimiter{\rrangle}{\mathclose}{mnlargesymbols}{'171}{mnlargesymbols}{'171}
\DeclareMathDelimiter{\lsem}{\mathopen}{mnlargesymbols}{'102}{mnlargesymbols}{'102}
\DeclareMathDelimiter{\rsem}{\mathclose}{mnlargesymbols}{'107}{mnlargesymbols}{'107}
\DeclareMathDelimiter{\langlebar}{\mathopen}{mnlargesymbols}{'152}{mnlargesymbols}{'152}
\DeclareMathDelimiter{\ranglebar}{\mathclose}{mnlargesymbols}{'157}{mnlargesymbols}{'157}
\DeclareMathDelimiter{\lWavy}{\mathopen}{mnlargesymbols}{'137}{mnlargesymbols}{'137}
\DeclareMathDelimiter{\rWavy}{\mathopen}{mnlargesymbols}{'137}{mnlargesymbols}{'137}

\newcommand{\chr}{\mathbf 1}
\newcommand{\abs}[1]{\lvert#1\lvert}
\newcommand{\norm}[1]{\lVert#1\lVert}

\newcommand{\FGamma}{\Gamma}
\newcommand{\FS}{\cF}\newcommand{\dG}{\sfd\FGamma}

\title[Lower Bound on the Critical Momentum of an Impurity in a BEC]{A Lower Bound on the Critical Momentum\\ of an Impurity in a Bose--Einstein Condensate}
\author{Benjamin Hinrichs}
\address{Benjamin Hinrichs, Universit\"at Paderborn, Institut f\"ur Mathematik, Institut f\"ur Photonische Quantensysteme, Warburger Str. 100, 33098 Paderborn, Germany}
\email{benjamin.hinrichs@math.upb.de}
\author{Jonas Lampart}
\address{Jonas Lampart, CNRS \& Laboratoire Interdisciplinaire Carnot de Bourgogne (UMR 6303), Université de Bourgogne Franche-Comté, 9 Av. A. Savary, 21078 Dijon Cedex, France.}
\email{jonas.lampart@u-bourgogne.fr}


\usepackage{xcolor}\usepackage{cancel}

\newcommand{\Pc}{\sfP\!_*}
\newcommand{\GL}{G_{\Lambda}}
\newcommand{\GLt}{\wt G_{\Lambda}}

\begin{document}

\begin{abstract} 
	\noindent
	In the Bogoliubov--Fröhlich model, we prove  that an impurity immersed in a Bose--Einstein condensate forms a stable quasi-particle when the total momentum is less than its mass times the speed of sound. The system thus exhibits superfluid behavior, as this quasi-particle does not experience friction.
	We do not assume any infrared or ultraviolet regularization of the model, which  contains massless excitations and  point-like interactions.
\end{abstract}

\maketitle

\section{Introduction}

An impurity in a Bose--Einstein condensate will create excitations out of the ground state and may form a quasi-particle, called the Bose polaron, consisting of the particle and a surrounding cloud of excitations.
The system is of interest in physics as the impurity can reveal properties of the condensate, such as superfluidity. Moreover, Bose--Einstein condensates are finely controllable experimental platforms from which one hopes to learn about polaron physics in solids by analogy.

The Bogoliubov--Fr\"ohlich Hamiltonian is an effective model for such a system, in which the particle is linearly coupled to Bogoliubov's excitation field.
This model is relevant if the interaction between the particle and the bosons is sufficiently weak to not significantly impact the condensate~\cite{GrusdtDemler.2016},
though there is some debate in the physics literature on what effects this model can or cannot capture~\cite{ChristensenLevinsenBruun.2015,GrusdtSchmidtShchadilovaDemler.2017}.
Recent mathematical results prove that it provides an accurate description of the system in certain mean-field~\cite{MysliwySeiringer.2020, LampartPickl.2022} and dilute~\cite{LampartTriay.2023} regimes.

In this letter we start from the translation-invariant Bogoliubov--Fr\"ohlich Hamiltonian in $\R^3$ and prove that the Bose polaron is stable when the total momentum is less than the impurity mass times the speed of sound. Mathematically, this corresponds to proving that the Hamiltonian at fixed total momentum has an eigenvalue at the bottom of its spectrum. Since the excitations in this model are massless, this eigenvalue is always embedded in the essential spectrum.
One expects that beyond some critical momentum this eigenvalue disappears and the system   exhibits a Cherenkov transition. That is, the polaron would radiate sound waves, thereby slowing down to a stable state of smaller momentum. This has been validated numerically in~\cite{Seetharametal.2021b,Seetharametal.2021}, but there does not seem to be a mathematical proof of such a statement.

The dichotomy of stability at small velocities and a friction effect at high velocities has been studied in a  model of a classical particle interacting with sound waves in the series of works~\cite{FrohlichGangSoffer.2011, FrohlichGangSoffer.2012, EgliGang.2012, Eglietal.2013, FrohlichGang.2014b, FrohlichGang.2014}, and later in~\cite{Leger.2020}. This model can be related to the Bose polaron system in a mean-field regime with a heavy impurity~\cite{Deckertetal.2014}.
A simplified model is obtained by decoupling the directions of propagation of the particle and the waves, which limits the back-reaction of the field on the particle. Such a model was studied in~\cite{BruneauDeBievre.2002} in the classical and~\cite{Bruneau.2007,DeBievreFaupinSchubnel.2017} in the quantum mechanical setting.

\subsection{The Bogoliubov--Fr\"ohlich Polaron}

The Bogoliubov--Fr\"ohlich Hamiltonian is characterized by the dispersion relation of the field of excitations, or phonons, and the form factor of the interaction.
The dispersion relation is
\begin{equation}\label{def:omega}
	\omega(k) \coloneqq \sqrt{\sfc^2 \abs{k}^2 + \xi^2\abs{k}^4},
\end{equation}
where $\sfc>0$ is the speed of sound and $\xi=1/(2m_\mathrm{B})$, for the mass $m_B$ of the bosons in the gas. The form factor of the particle-phonon interaction is
\begin{equation}\label{def:v}
	v_\Lambda(k) \coloneqq \sfg\chr_{\abs k <\Lambda} \sqrt{\abs{k}^2/\omega(k)},
\end{equation}
where $\Lambda$ is an ultraviolet cutoff (that may take the value infinity) and
$\sfg$ is a coupling constant, whose value will not be important in our analysis.

Our model is then realized as a selfadjoint lower-semibounded Hamiltonian on the bosonic Fock space $\FS$ over $L^2(\IR^3)$. We use the standard notation $a_k$, $a_k^*$ for the creation and annihilation operators on $\FS$ in the sense of operator-valued distributions.
As usual, writing $\dG(f) = \int f(k)a_k^*a_k\d k$ and $\ph(g) = \int (g(k)^*a_k + g(k)a_k^*)\d k$ for second quantization and field operators, respectively,
we define the Bogoliubov--Fr\"ohlich Hamiltonian at momentum $P\in\IR^3$ with cutoff $\Lambda<\infty$ by
\begin{equation}\label{def:HL}
	H_\Lambda(P)\coloneqq \tfrac 12(P-\dG(\hat p))^2 + \dG(\omega) + \ph(v_\Lambda),
\end{equation}
where $\wh p=(\wh p_1,\wh p_2, \wh p_3)$ denotes the vector of multiplication operators on $L^2(\IR^3)$ given by $\wh p_i f(k) = k_if(k)$.
Note that we have set the impurity mass equal to one, keeping $\sfc$, $\xi$ and $\sfg$ as the model parameters. Moreover, we may choose $\sfg\ge0$ without loss of generality, as the models with different signs (phases) in the coupling are unitarily equivalent via $e^{i\mathrm{arg}(\sfg) \dG(1) }$.

By the Kato--Rellich theorem, it follows from standard estimates that $H_\Lambda(P)$ is a selfadjoint lower-semibounded operator with
 $\cD(H_\Lambda(P))=\cD(H_0(0))=\cD(\dG(\hat p)^2)\cap\cD(\dG(\omega)$ for all $\Lambda<\infty$, since $v_\Lambda,\omega^{-1/2}v_\Lambda\in L^2(\IR^3)$.

For $\Lambda=\infty$, $H(P)$ is defined by the following renormalization result from~\cite{Lampart.2019}. For the convenience of the reader, we sketch the proof in \cref{sec:renormalization}.
\begin{prop}\label{prop:ren}
	There exist $(\Sigma_\Lambda)_{\Lambda\ge 0}\subset \IR$  and, for all $P\in\IR^3$, a selfadjoint lower-semibounded operator $H(P)$ (given in \cref{thm:ren}) such that $H_\Lambda(P)-\Sigma_\Lambda$ converges to $H(P)$ as $\Lambda\to\infty$ in the norm resolvent sense.
\end{prop}
\begin{proof}
	The statement is that of \cref{thm:ren} with $\kappa=0$.
\end{proof}
We are interested in studying the critical momentum of the operator $H(P)$.

\subsection{Critical Momentum}
As described earlier in this introduction, the polaron may become unstable for large momentum. We define the the critical momentum as
\begin{equation}
	\Pc \coloneqq \sup \{P\in\IR^3 \colon \inf \sigma (H(P))\ \mbox{is an eigenvalue of}\ H(P)   \}.
\end{equation}
The main result of this article can now easily be stated.
\begin{thm}\label{thm:critmom}
	For any coupling constant $\sfg\ge 0$ and any speed of sound $\sfc> 0$, we have
	\[\Pc \ge \sfc. \]
\end{thm}
\begin{proof}
	The statement is an immediate corollary of \cref{thm:gsex}.
\end{proof}
\begin{rem}
	It would be interesting to show that $\Pc$ is finite, and that there is a unique transition, i.e., $\inf \sigma (H(P))$ is an eigenvalue if and only if $P<\Pc$. The numerical study~\cite{Seetharametal.2021b} supports this picture and indicates that the second derivative of $E$ has a jump at $\Pc$.
	Moreover,~\cite[Fig.3(e)]{Seetharametal.2021b} suggests that $\Pc$ increases from $\sfc$ to infinity as $\sfg\to \infty$.
\end{rem}

\begin{rem}
An interpretation of the statement is to think of $\Pc=m_*\sfc$, where $m_*$ is the effective mass of the polaron (compare~\cite{Seetharametal.2021}). Then we have shown $m_*\ge1$, meaning that the quasi-particle is heavier than the impurity of mass one,  in agreement with the picure that the particle is dressed by a cloud of phonons, increasing its effective mass.
This definition of an effective mass is, of course, different from the common definition by the curvature $\partial_{\abs P}^2 E(P)|_{P=0}$ at zero, but one may still expect similar qualitative behavior, see~\cite{Seiringer.2021, DybalskiSpohn.2020} for a discussion of the latter quantity in the Fr\"ohlich polaron model.
\end{rem}

In order to prove Theorem~\ref{thm:critmom}, we have to deal with both an ultraviolet and an infrared problem. The first is due to the fact that $v_\infty\notin L^2(\R^3)$ (and also $\omega^{-1/2}v_\infty \notin  L^2(\R^3)$).
Using the method of interior boundary conditions, we can nevertheless describe the Bogoliubov--Fr\"ohlich Hamiltonian, in particular its domain, without any ultraviolet regularization. This goes back to a recent article by the second author \cite{Lampart.2019}, building on techniques developed for the related Fr\"ohlich and Nelson models in \cite{LampartSchmidt.2019} and improved on in the subsequent articles \cite{Schmidt.2019,Lampart.2019b,Schmidt.2021,Lampart.2023}.

The infrared problem is due to the fact that $\omega(0)=0$, which entails that $H(P)$ does not have a spectral gap and $\inf \sigma(H(P))=\inf \sigma_\mathrm{ess}(H(P))$.
For massive polaron models, i.e., models satisfying $\essinf \omega>0$, the existence of ground states is well known, see for example \cite{Frohlich.1973,Spohn.1988,DerezinskiGerard.1999}.
In the massless case, one distinguishes between the infrared singular case $v/\omega\notin L^2(\IR^3)$ and the infrared regular one $v/\omega\in L^2(\IR^3)$.
In the infrared singular case, e.g., the famous Nelson model \cite{Nelson.1964}, absence of ground states at arbitrary total momentum (and all non-zero couplings) has been shown, cf. \cite{Frohlich.1973,HaslerHerbst.2008a,Dam.2018,DamHinrichs.2021}. In our case, however, the model is infrared-regular, as can be easily checked.
There exists a variety of perturbative methods to prove existence of ground states in such a case for small values of the total momentum and the coupling constant, e.g., 
operator theoretic renormalization \cite{BachFroehlichSigal.1998b}, iterated perturbation theory \cite{Pizzo.2003,DybalskiPizzo.2014} and functional integration methods \cite{Spohn.1998,BetzHiroshimaLorincziMinlosSpohn.2002}.
In particular, in presence of an ultraviolet cutoff $\Lambda<\infty$ and for small coupling, the fact that $\Pc>0$ follows from~\cite[Proposition~1.1]{DybalskiPizzo.2014}.
Hence, we extend the existence of a ground state to the case without ultraviolet cutoff, arbitrary coupling and a larger set of total momenta.
The method we use in our proof is an adaption of a compactness argument first applied in \cite{GriesemerLiebLoss.2001} and subsequently employed in the study of various models, e.g., the spin boson model \cite{HaslerHinrichsSiebert.2021a}, the Nelson model \cite{HiroshimaMatte.2019,HaslerHinrichsSiebert.2023} and the Pauli--Fierz model \cite{Matte.2016,HaslerSiebert.2020}.
The general strategy is to introduce an artificial boson mass $\kappa>0$, and then prove that the set of ground states with $\kappa\to 0$ is pre-compact and provides a minimizing sequence for $H(P)$.

In the remainder of this letter,
we sketch the renormalization procedure leading to \cref{prop:ren} in \cref{sec:renormalization} and give the proof of \cref{thm:critmom} in \cref{sec:existence}.


\section{Renormalization and Properties of the Bogoliubov--Fr\"ohlich Polaron}\label{sec:renormalization}
In this \lcnamecref{sec:renormalization}, we sketch the proof of \cref{prop:ren}, by
reviewing the renormalization method employed in \cite{Lampart.2020}.
The key idea is to identify a divergent and $P$-independent contribution $\Sigma_\Lambda$ to $ \inf \sigma (H_\Lambda(P))$. This contribution is of the form
\begin{equation}
 \Sigma_\Lambda = e_1 \Lambda + e_2  \log \Lambda + \mathcal{O}(1).
\end{equation}
The two divergent contributions of different orders arise in a two-step procedure of rewriting $H_\Lambda(P)$.

Throughout this section, we assume $P\in\IR^3$ to be fixed.
We emphasize that most of the defined objects, except for the contributions to $\Sigma_\Lambda$, do have a $P$-dependence.
We now fix some parameter $\mu>0$ and define
\begin{equation}
 \GL = - \big( a(v_\Lambda) (H_0(P)+\mu)^{-1}\big)^*.
\end{equation}
Employing that $\omega^{-1}v\in L^2(\IR^3)$, one can show that $\GL$ is a bounded operator, including the case $\Lambda=\infty$, see the proof of \cref{thm:ren} below for more details.
Further, for $\Lambda<\infty$, we have the simple identity
\begin{equation}\label{eq:G-identity}
 H_\Lambda(P)=(1-\GL^*)(H_0(P)+\mu)(1-\GL) - \GL^*(H_0(P)+\mu)\GL - \mu,
\end{equation}
which follows by expanding the product.
The first singular contribution is contained in the term
\begin{equation}
 -    \GL^*(H_0(P)+\mu)\GL = - a(v_\Lambda)(H_0(P)+\mu)^{-1}a^*(v_\Lambda)
\end{equation}
To make it explicit, we will put the creation and annihilation operators in this expression in normal order.
With the pull-through formula $a_k(H_0(P)+\mu)^{-1}=(H_0(P-k)+\omega(k)+\mu)^{-1}a_k$,
which holds by inspection on every $n$-particle sector of $\cF$ (see for example \cite[Lemma~IV.8]{BachFroehlichSigal.1998b}), we find
\begin{align}
  &- a(v_\Lambda)(H_0(P)+\mu)^{-1}a^*(v_\Lambda)  = - \int a_k \frac{v_\Lambda(k) v_\Lambda(\ell)}{H_0(P)+\mu}a_\ell^*\d k \d \ell  \\
  &\qquad= - \int \frac{v_\Lambda(k)^2}{H_0(P-k)+\omega(k)+\mu}\d k - \int a_\ell^* \frac{v_\Lambda(k) v_\Lambda(\ell)}{H_0(P-k-\ell)+\omega(k)+\omega(\ell)+\mu}a_k \d k \d\ell.\notag
\end{align}
With this order of $a^*_\ell, a_k$, the second term will be well defined also for $\Lambda=\infty$ as an unbounded operator, since the decay of   $a_k\Psi$  in $k$ for an element $\Psi$ of its domain will make the integral convergent.
For the first term this is not the case, and we will need to first subtract its divergent contribution to take $\Lambda\to \infty$. This can be chosen as
\begin{equation}\label{eq:SL1}
 \Sigma^{(1)}_\Lambda= - \int \frac{v_\Lambda(k)^2}{\frac12k^2+\omega(k)}\d k,
\end{equation}
which has a divergence proportional to $\Lambda$ since $v_\Lambda(k)$ is of order one for large $k$.
We then define $T_{\Lambda,1}=\Theta_{\Lambda,1,0}+\Theta_{\Lambda,1,1}$, where $\Theta_{\Lambda,1,0}=\theta_{\Lambda,1,0}(\dG(\hat p), \dG(\omega))$ is a multiplication operator in the momentum representation and $\Theta_{\Lambda,1,1}= \int a_\ell^*\theta_{\Lambda,1,1}(\dG(\hat p), \dG(\omega),k,\ell)a_k\d k\d l$ is an integral operator, with
\begin{equation}\label{eq:theta-def}
\begin{aligned}
                 \theta_{\Lambda,1,0} (p, \eta)&= - \int \bigg(\frac{v_\Lambda(k)^2}{\frac12(P-p-k)^2 + \eta+\omega(k)+\mu}-\frac{v_\Lambda(k)^2}{\frac12k^2+\omega(k)}\bigg) \d k \\
 \theta_{\Lambda,1,1}(p, \eta,k,\ell) &=-  \frac{v_\Lambda(k) v_\Lambda(\ell)}{\frac12(P-p-k-\ell)^2 + \eta+\omega(k)+\omega(\ell)+\mu}.
                \end{aligned}
 \end{equation}
Now, $T_\Lambda$ makes sense also for $\Lambda=\infty$, since the integral defining $\theta_{\Lambda,1,0}$ has a limit for $\Lambda\to\infty$, and hence we could try to employ the identity
\begin{equation}
	\GL^*(H_0(P)+\mu)\GL - \Sigma^{(1)}_\Lambda = T_\Lambda
\end{equation}
to define the second term in \cref{eq:G-identity} for a definition of $H(P)$.
This is not enough, however, to remove the cutoff completely, since the (form) domain of the first term $(1-G_\infty^*)(H_0(P)+\mu)(1-G_\infty)$ is not contained in the (form) domain of $T_\infty$. To remedy this issue, we include $T_\Lambda$ with the free operator.

For $\Lambda\in \R_+\cup\{\infty\}$, let
\begin{equation}\label{eq:G-tilde-def}
 \GLt = - \big( a(v_\Lambda) (H_0(P)+T_\Lambda+\mu)^{-1}\big)^*.
\end{equation}
Then we can write a similar identity to \eqref{eq:G-identity} for $\Lambda<\infty$, explicitly
\begin{equation}\label{eq:Gtilde-identity}
 H_\Lambda =  (1-\GLt^*)(H_0(P)+T_\Lambda+\mu)(1- \GLt) -a(v_\Lambda)(H_0(P)+T_\Lambda+\mu)^{-1}a^*(v_\Lambda)-T_\Lambda - \mu.
\end{equation}
Expanding the resolvent gives
\begin{align}
&- a(v_\Lambda)(H_0(P)+T_\Lambda+\mu)^{-1}a^*(v_\Lambda)-T_\Lambda \\
&\quad= \Sigma_\Lambda^{(1)} + a(v_\Lambda)(H_0(P)+T_\Lambda+\mu)^{-1}T_\Lambda(H_0(P)+\mu)^{-1}a^*(v_\Lambda)\notag \\
&\quad  = \Sigma_\Lambda^{(1)} + a(v_\Lambda)(H_0(P)+\mu)^{-1}T_\Lambda(H_0(P)+\mu)^{-1}a^*(v_\Lambda)  \notag \\
&\qquad - a(v_\Lambda)(H_0(P)+\mu)^{-1}T_\Lambda(H_0(P)+T_\Lambda+\mu)^{-1}T_\Lambda(H_0(P)+\mu)^{-1}a^*(v_\Lambda). \notag
\end{align}
The term in the last line is regular in the case $\Lambda=\infty$ and will be treated as a remainder, while the first still contains the logarithmic divergence. To extract this, we proceed as before and put the creation and annihilation operators in normal order. However, there is now also the possibility of picking up a commutator between the operators in $\Theta_{\Lambda,1,1}$ and the outer creation/annihilation operators. With this in mind, the term with no remaining creation and annihilation operators reads
\begin{align*}
 \int  \frac{v_\Lambda(k)^2 \theta_{\Lambda,1,0}(\dG(\hat p)+k, \dG(\omega)+\omega(k))}{(H_0(P-k)+\omega(k)+\mu)^2}\d k - \int \frac{v_\Lambda(k) v_\Lambda(\ell)\theta_{\Lambda,1,1}(\dG(\hat p), \dG(\omega),k,\ell)}{(H_0(P-k)+\omega(k)+\mu)(H_0(P-\ell)+\omega(\ell)+\mu)}\d k\d \ell.
\end{align*}
These integrals have a logarithmic divergence as $\Lambda\to \infty$, captured by
\begin{equation}\label{eq:SL2}
	\begin{aligned}
		 \Sigma_\Lambda^{(2)} = &\int  \frac{v_\Lambda(k)^2 \theta_{\Lambda,1,0}(k, \omega(k))}{(\frac12k^2+\omega(k))^2}\d k  \\ &- \int \frac{v_\Lambda(k)^2 v_\Lambda(\ell)^2}{(\frac12k^2+\omega(k))(\frac12(k+\ell)^2 + \omega(k)+ \omega(\ell))(\frac12\ell^2+\omega(\ell))}\d k\d \ell.
	\end{aligned}
\end{equation}
After subtracting this, we define $\widetilde T_{\Lambda}=\Theta_{\Lambda,2,0} + \Theta_{\Lambda,2,1} +\Theta_{\Lambda,2,2}$, where $\Theta_{\Lambda,2,0}$ is a multiplication operator of the same type as $\Theta_{\Lambda,1,0}$, and $\Theta_{\Lambda,2,1}, \Theta_{\Lambda,2,2}$ are integrals with one, respectively two, remaining creation and annihilation operators.  The expression for $\Theta_{\Lambda,2,0}$ is given by
\begin{align}\label{eq:theta2-0}
 \theta_{\Lambda,2,0}(p,\eta) &= \int  \frac{v_\Lambda(k)^2 \theta_{\Lambda,1,0}(p+k, \eta+\omega(k))}{(\frac12(P-p-k)^2+ \eta+\omega(k)+\mu)^2} \d l \\
 &\qquad + \int \frac{v_\Lambda(k) v_\Lambda(\ell)\theta_{\Lambda,1,1}(p, \eta,k,\ell)}{(\frac12(P-p-k)^2+ \eta+\omega(k)+\mu)(\frac12(P-p-\ell)^2+ \eta+\omega(\ell)+\mu)} \d k\d\ell - \Sigma_\Lambda^{(2)}, \notag
\end{align}
where we observe that $\Sigma_\Lambda^{(2)}$ is simply the value of the integrals at $P=p=\eta=\mu=0$.
The integral operators have the kernels
\begin{align}
 \theta_{\Lambda,2,1}(p,\eta,k,\ell) &= \frac{v_\Lambda(k)v_\Lambda(\ell) \theta_{\Lambda,1,0}(p+k+\ell, \eta+\omega(k)+\omega(\ell))}{(\frac12(P-p-k)^2+ \eta+\omega(k)+\mu)(\frac12(P-p-\ell)^2+ \eta+\omega(\ell)+\mu)} \\
 &\qquad + \int \frac{v_\Lambda(\xi)^2\theta_{\Lambda,1,1}(p+\xi, \eta+\omega(\xi),k,\ell)}{(\frac12(P-p-\xi)^2+ \eta+\omega(\xi)+\mu)^2}\d k\d l\notag \\
 \theta_{\Lambda,2,2}(p,\eta,k_1, k_2,\ell_1, \ell_2)&=  \frac{v_\Lambda(k_1) v_\Lambda(\ell_1)\theta_{\Lambda,1,1}(p+k_1+\ell_1, \eta+\omega(k_1)+\omega(\ell_1),k_2,\ell_2)}{(\frac12(P-p-k_1)^2+ \eta+\omega(k_1)+\mu)(\frac12(P-p-\ell_1)^2+ \eta+\omega(\ell_1)+\mu)}.\label{eq:theta2-1}
\end{align}
Again, the definition of $\widetilde T_\Lambda$ may be extended to $\Lambda=\infty$ since these functions are defined also for this value.
Finally, the definition of the remainder term reads
\begin{equation}\label{eq:R-def}
 \begin{aligned}
 	R_\Lambda &= - a(v_\Lambda)(H_0(P)+\mu)^{-1}T_\Lambda(H_0(P)+T_\Lambda+\mu)^{-1}T_\Lambda(H_0(P)+\mu)^{-1}a^*(v_\Lambda)
 	\\& 
 	= G_\Lambda^* T_\Lambda (H_0(P)+T_\Lambda+\mu)^{-1}T_\Lambda G_\Lambda
 	,
 \end{aligned}
\end{equation}
which defines a bounded operator also for $\Lambda=\infty$.

Since we require an infrared regularization in \cref{sec:existence} additionally to the ultraviolet one provided by the cutoff $\Lambda$, we directly consider the family of operators $H_{\kappa,\Lambda}$ given by \cref{def:HL} with $\omega$ replaced by $\omega_\kappa=\omega+\kappa$, i.e.,
\begin{equation}\label{def:HkL}
	H_{\kappa,\Lambda}(P) \coloneqq H_\Lambda(P) + \kappa N - \Sigma_\Lambda^{(1)} - \Sigma_\Lambda^{(2)}\qquad \mbox{for}\ \kappa\ge0,\ \Lambda\in\R_+,
\end{equation}
where $N=\dG(1)$ is the particle number operator as usual
and we incorporated the ultraviolet
renormalization, by directly subtracting the divergent energy contributions as defined in \cref{eq:SL1,eq:SL2}.
Note that $\cD(H_{\kappa,0}(P))= \cD(H_{0,0}(0))\cap \cD(\kappa N)$, so we simply denote this domain by $\cD(H_{\kappa,0})$. For $\Lambda<\infty$, \cref{def:HkL} immediately defines a selfadjoint lower-semibounded operator on $\cD(H_{\kappa,0})$, since $\omega^{-1/2}v_\Lambda\in L^2(\IR^3)$.

The preceding discussion applies in the same way with $\omega_\kappa$, yielding objects $T_\Lambda=T_{\kappa,\Lambda}$, $\widetilde G_\Lambda=\widetilde G_{\kappa,\Lambda}$, $\widetilde T_\Lambda=\widetilde T_{\kappa,\Lambda}$, whose dependence on $\kappa$ we will not make explicit.
\cref{prop:ren} is now a consequence of the following theorem for $\kappa=0$.

\begin{thm}[\cite{Lampart.2020}]\label{thm:ren}
 Let $\kappa\ge 0$ and let $\widetilde G_\infty$, $T_\infty$, $\widetilde T_\infty$, $R_\infty$ be defined by~\eqref{eq:G-tilde-def},~\eqref{eq:theta-def}, ~\eqref{eq:theta2-0}--\eqref{eq:theta2-1}, and~\eqref{eq:R-def} with $\omega=\omega_\kappa$ respectively.
 The operator
 \begin{align*}
  H_{\kappa,\infty}(P)&= (1-\widetilde G_\infty^*)(H_{\kappa,0}(P)+T_\infty+\mu)(1-\widetilde G_\infty) + \widetilde T_\infty + R_\infty -\mu\\
  \cD(H_{\kappa,\infty}(P))&= \big\{ \psi \in \cF : (1-\widetilde G_\infty)\psi\in \cD(H_{\kappa,0})\big\}
 \end{align*}
is selfadjoint and bounded from below.
We have the convergence
\[
 H_{\kappa,\Lambda}(P) \to H_{\kappa,\infty}(P)
\]
in norm resolvent sense.
\end{thm}
\begin{proof}[Sketch of the proof]
We give a short outline of the proof with references to key technical lemmas for the convenience of the reader.
Throughout this proof, $\kappa\ge0$ is fixed.

The first step is to prove that
\begin{align}\label{eq:Tbounds}
 \begin{aligned}
 	&\|T_\Lambda \psi\|\le C \|(H_{\kappa,0}(P)+1)^{1/2}\psi\| \qquad \mbox{for}\ \Lambda\in\R_+\cup\{\infty\}, \\
 	&\|(T_\Lambda-T_\infty)\psi\| \le C_\Lambda \|(H_{\kappa,0}(P)+1)^{1/2+\eps}\psi\|,
 \end{aligned}
\end{align}
with $\eps>0$ and $\lim_{\Lambda\to \infty} C_\Lambda=0$ (the part~$\Theta_{\Lambda, 1,0}$ can be bounded by an elementary calculation; concerning $\Theta_{\Lambda,1,1}$, see~\cite[Lem.17]{Lampart.2019b} and~\cite[Lemma B.2]{Lampart.2023} for proofs in the case $\kappa>0$ that are easily adapted to $\kappa=0$).

Using this,  one shows that  for $\Lambda\in \R_+ \cup\{\infty\}$, $\GLt$ are bounded operators on $\cF$, satisfying 
\begin{equation}
 \| (H_{\kappa,0}(P)+\mu)^s \GLt \| \le C,\qquad   (H_{\kappa,0}(P)+\mu)^s(\GLt-\widetilde G_\infty) \stackrel{\Lambda\to \infty}{\longrightarrow} 0 \qquad\text{for $0\le s<1/4$}.
\end{equation}
This follows easily from the bound $\|a(f)\dG(\eta)^{-1/2}\|\le \|f/\eta\|$, $v/\omega\in L^2(\IR^d)$ and the fact that $T_\Lambda$ is an infinitesimal perturbation of $H_{\kappa,0}(P)$.

In particular, for $\mu$ large enough, $ \|\GLt\|<1$, so $1-\GLt$ is invertible. This shows that $\cD(H_{\kappa,\infty}(P))$ is dense and combined with \cref{eq:Tbounds}, the operator
\begin{equation}
 K:=(1-\widetilde G_\infty^*)(H_{\kappa,0}(P)+T_\infty+\mu)(1-\widetilde G_\infty)
\end{equation}
is selfadjoint and bounded from below on this domain.
The terms $\widetilde T_\infty$, $R_\infty$ will be treated as perturbations of $K$. For $R_\infty$, boundedness follows directly from the properties of $T_\Lambda$ and $\widetilde G_\Lambda$.

For $\widetilde T_\Lambda$ one can again show 
\begin{align}
 \|\widetilde T_\Lambda \psi\|&\le C \|(H_{\kappa,0}(P)+1)^{\eps}\psi\|, \\
 \|(\widetilde T_\Lambda-\widetilde T_\infty)\psi\| &\le C_\Lambda \|(H_{\kappa,0}(P)+1)^{\eps}\psi\|
\end{align}
for $\eps>0$ and $\lim_{\Lambda\to \infty} C_\Lambda=0$ (cf.~\cite[Lemma B.2]{Lampart.2023},~\cite[Lem.19]{Lampart.2019b}).
This implies that 
\begin{align}
 \| \widetilde T_\infty \psi \| &\le \| \widetilde T_\infty (1-\widetilde G_\infty)\psi \| + \| \widetilde T_\infty \widetilde G_\infty \psi \|  \\
 &\le C( \|   (H_{\kappa,0}(P)+1)^{\eps} (1-\widetilde G_\infty)\psi\| + \|  (H_{\kappa,0}(P)+1)^{\eps} \widetilde G_\infty\psi\|) \notag \\
 &\le \delta \|  K_\kappa \psi\| + C_\delta \|\psi\|\notag
\end{align}
for any $\delta>0$. Thus $H_{\kappa,\infty}(P)$ is selfadjoint by the Kato--Rellich theorem.

Convergence of resolvents follows from the identity~\eqref{eq:Gtilde-identity}, the resolvent formula and the convergence properties of $T_\Lambda$, $\GLt$ already mentioned. 
\end{proof}
From the proof, we also obtain the following \cref{lem:core}, which relates the domains of $H(P)$ and $N$.
It will be important to our proof of \cref{thm:critmom} in the next \lcnamecref{sec:existence}.
\begin{lem}\label{lem:core}
	For any $P\in\IR^3$,
	the subspace $\cD(N)\cap \cD(H_{0,\infty}(P))$ is a core for $H_{0,\infty}(P)$.
	Further,   $\cD(H_{\kappa,\infty}(P)) = \cD(N)\cap \cD(H_{0,\infty}(P))$ for all $\kappa>0$ and
	$H_{\kappa,\infty}(P) = H_{0,\infty}(P)+\kappa N$.
\end{lem}
\begin{proof}
From \cref{thm:ren}, we know that
$\cD(H_{\kappa,\infty}(P))=(1-\widetilde G_\infty)^{-1} \cD(H_{\kappa,0})$.
 Moreover, for any core $\cC$ of $H_{\kappa,0}(P)$, $(1-\widetilde G_\infty)^{-1}\cC$ is a core for $H_{\kappa,\infty}(P)$, since $ (1-\widetilde G_\infty)^{-1} : \cD(H_{\kappa,0})\to \cD(K)$ is continuous for the graph norms.
 Hence, to prove the domain statements, it suffices to show $(1-\wt G_\infty)\cD(N) = \cD(N)$.
 This follows from the observation $N\wt G_\infty = \wt G_\infty (N+1)$, implying
 \begin{align}
 	&\|N(1-\wt G_\infty)\psi\| \le \|N\psi\| + \|\wt G_\infty\|\|(N+1)\psi\|,
 	\\
 	&\| N \psi\| \le \| N(1-\widetilde G_\infty)\psi\| + \|N \widetilde G_\infty \psi\| \le \| N(1-\widetilde G_\infty)\psi\| + \|\widetilde G_\infty\| \| (N+1) \psi\|,
 \end{align}
 from where we conclude using $\|\widetilde G_\infty\|<1$.
 Moreover,
 \begin{equation}
  H_{\kappa,\infty}(P) = H_{0,\infty}(P)+\kappa N
 \end{equation}
holds since both sides are the weak graph limit of $H_{\kappa,\Lambda}(P)$, by \cref{thm:ren}, \cite[Thm.VIII.26]{ReedSimon.1972}, and the uniform bound $\|N\psi\|\le C \|(H_{\kappa,\Lambda}(P)+\mu)\psi\|$.
\end{proof}

\begin{rem}
One can observe that $\cD(H(P))\neq \cD(H(P')$ for $P\neq P'$.
This is the case because
\begin{equation}
(H_0(P)+\mu)^{-1}\dG(\hat p)G_\infty,
\end{equation}
which is proportional to the difference of $G_\infty$ for two different values of $P$, does not map $\cD(H_0(0))$ to itself. It does, however, map the form domain of $H_{\kappa,0}(P)$ to itself, so the operators with different total momenta still have comparable quadratic forms. Notwithstanding, this fact will not be used in our arguments.
\end{rem}


\section{Existence of Ground States}\label{sec:existence}
In this \lcnamecref{sec:existence}, we prove the following \lcnamecref{thm:gsex}.
\begin{thm}\label{thm:gsex}
	If $\abs P<  \sfc$, then $\inf \sigma(H(P))$ is an eigenvalue of $H(P)$.
\end{thm}
To prove the statement, we approximate $H(P)=H_{0,\infty}(P)$ by the infrared regularized Hamiltonians $H_{\kappa,\infty}(P)$ with $\kappa>0$ introduced in \cref{def:HkL}.
Further, we write
\begin{equation}
	E_{\kappa,\Lambda}(P) \coloneqq \inf \sigma(H_{\kappa,\Lambda}(P)) \qquad\text{for all $\kappa\ge0$, $\Lambda\in\R_+\cup\{\infty\}$, $P\in\IR^3$}.
\end{equation}
Let us first observe that the ground state energies converge, when removing any regularization.
\begin{lem}\label{lem:gsconv}
	For any fixed $\kappa\ge 0$, $\Lambda\in\R_+\cup\{\infty\}$ and $P\in\IR^3$, we have
	\[ E_{\kappa,\infty} = \lim_{\sigma\to \infty}E_{\kappa,\sigma}
		\qquad\mbox{and}\qquad
		E_{0,\Lambda} = \lim_{\eta\downarrow 0} E_{\eta,\Lambda}.
	 \]
\end{lem}
\begin{proof}
	The first statement is a consequence of the norm resolvent convergence established in \cref{thm:ren} (cf. \cite{Oliveira.2009}).
	 For the second statement, we observe that $\cD(N)\cap \cD(H_{0,\Lambda}(P))$ is a core for $H_{0,\Lambda}(P)$, by the Kato--Rellich theorem for $\Lambda<\infty$ and by \cref{lem:core} for $\Lambda=\infty$. Hence, picking any $\eps>0$, there exists $\ph_\eps\in\cD(N)\cap \cD(H_{0,\Lambda}(P))$ with $\|\ph_\eps\|=1$ such that $\braket{\ph_\eps,H_{0,\Lambda}(P)\ph_\eps}<E_{0,\Lambda}(P)+\eps$.
	Further employing that $H_{\eta,\Lambda}(P)-H_{0,\Lambda}(P)\ge 0$ (as a form inequality), again by \cref{def:HkL,lem:core}, we find
	\begin{align*}
		E_{0,\Lambda}(P) \le E_{\eta,\Lambda}(P) \le \braket{\ph_\eps,H_{\eta,\Lambda}(P)\ph_\eps} = \braket{\ph_\eps,H_{0,\Lambda}(P)\ph_\eps} + \eta\braket{\ph_\eps,N\ph_\eps} \le E_{0,\Lambda} + \eps + \eta\braket{\ph_\eps,N\ph_\eps}.
	\end{align*}
	First taking $\eta\downarrow0$ and then $\eps\downarrow 0$ finishes the proof.
\end{proof}
The mass term $\kappa N$ ensures the existence of a spectral gap for small enough $P$, as a consequence of the following well-known HVZ-type theorem \cite{Frohlich.1973,Moller.2005}.
\begin{prop}[{\cite[Theorem 1.2]{Moller.2005}}]\label{prop:HVZ}
	For all $\kappa> 0$, $\Lambda\in\R_+$, we have
		\begin{equation}\label{eq:HVZ}
		\inf \sigma_{\sfe\sfs\sfs}(H_{\kappa,\Lambda}(P)) = \inf_{\substack{k_1,\ldots,k_n\in\IR^3\\n\in\IN}} \big[E_{\kappa,\Lambda}(P-k_1-\cdots-k_n) + \omega(k_1)+\cdots \omega(k_n) + n\kappa\big].
	\end{equation}
\end{prop}
In view of the above \lcnamecref{prop:HVZ}, we need to estimate the difference $E(P-k)-E(P)$.
This can be done using simple convexity arguments, cf. \cite{LossMiyaoSpohn.2007,HaslerHinrichsSiebert.2023}.
\begin{lem}\label{lem:convex}
	Let $\kappa\ge 0$, $\Lambda\in\R_+\cup\{\infty\}$ and $P,K\in\IR^3$. Then
	\begin{equation*}
		E_{\kappa,\Lambda}(P-K)-E_{\kappa,\Lambda}(P)
		\ge - |K||P|.
	\end{equation*}
\end{lem}
\begin{proof}
	First, we assume that $\kappa>0$ and $\Lambda<\infty$ and prove the inequalities
	\begin{equation}\label{eq:massshellinequ}
		0 \le E_{\kappa,\Lambda}(P) - E_{\kappa,\Lambda}(0) \le  \tfrac 12 \abs P^2 \quad\text{for all $P\in\IR^3$.}
	\end{equation}
	 The first inequality goes back to Gross \cite{Gross.1972}, see \cite[Lemma~3.4]{Hinrichs.2022} for a recent adaption which covers our case.
	 Now, note that \cref{prop:HVZ} combined with the first inequality
	 yields
	 \begin{equation}
	 \inf \sigma_\mathrm{ess}(H_{\kappa,\Lambda}(0)) \ge E_{\kappa,\Lambda}(0) + \kappa,
	 \end{equation}
so $E_{\kappa,\Lambda}(0) $ is a discrete eigenvalue with corresponding normalized eigenvector $\psi_0\in\cD(H_{\kappa,\Lambda}(0))=\cD(H_{\kappa,\Lambda}(P))$. Then
	 \[ E_{\kappa,\Lambda}((P)) \le \braket{\psi_0,H_{\kappa,\Lambda}(P)\psi_0} = E_{\kappa,\Lambda}(0) + \tfrac 12\abs P^2  + \braket{\psi_0,P\cdot \dG(\hat p) \psi_0}. \]
	This implies $\tfrac 12\abs P^2  + \braket{\psi_0,P\cdot \dG(\hat p) \psi_0}\ge 0$ for all $P\in\IR^3$. Letting $P\to 0$, this yields $e\cdot\braket{\psi_0,\dG(\hat p)\psi_0} \ge 0$ for all normalized $e\in\IR^3$,
	whence $\braket{\psi_0,\dG(\hat p)\psi_0}=0$.
	This proves the upper bound in \cref{eq:massshellinequ}.
	 
	Clearly, the map
	\begin{equation}
	P\mapsto \tfrac 12P^2-E_{\kappa,\Lambda}(P) = - \inf_{\psi\in \cD(H_{\kappa,\Lambda}(P))} \langle \psi,  (-P\cdot \dG(\hat p)+ \tfrac12\dG(\hat p)^2 + H_{\kappa,\Lambda}(0) )\psi\rangle
	\end{equation}
is convex, as a supremum over linear functions of $P$.
	By a general result on convex functions taking nonnegative values below the standard parabola (essentially the fact that such a function must lie below any segment that intersects its graph and is tangent to the parabola, cf. \cite[Appendix~A]{LossMiyaoSpohn.2007} or \cite[Cor.~A.6]{HaslerHinrichsSiebert.2023}), this gives for $\kappa>0$, $\Lambda<\infty$
	\begin{equation}
		E_{\kappa,\Lambda}(P-K)-E_{\kappa,\Lambda}(P)
		\ge
		\begin{cases}
			- \abs K \abs P +\frac 12 \abs K^2 & \mbox{if}\ \abs K\le \abs P,\\
			-\frac 12 \abs P^2 & \mbox{if}\ \abs K>\abs P.
		\end{cases}
	\end{equation}
	In both cases the right hand side is larger than $-|K||P|$ as claimed.
	This proves the claim for $\kappa>0$ and $\Lambda<\infty$.
	The general statement follows from the convergence results of \cref{lem:gsconv}.
\end{proof}

\begin{cor}\label{cor:massiveexistence}
	If $\abs P\le \sfc$, then $E_{\kappa,\Lambda}(P)$ is a discrete eigenvalue of $H_{\kappa,\Lambda}(P)$ for all $\kappa>0$ and $\Lambda\in\R_+\cup\{\infty\}$.
\end{cor}
\begin{proof}
	First assume $\Lambda<\infty$.
	Combining  the HVZ Theorem,~\cref{prop:HVZ}, with  \cref{lem:convex} gives
	\begin{equation}\label{eq:E-difference}
	 \inf \sigma_{\sfe\sfs\sfs}(H_{\kappa,\Lambda}(P)) - E_{\kappa,\Lambda}(P) \ge
	  \inf_{\substack{k_1,\ldots,k_n\in\IR^3\\n\in\IN}} \bigg(\sum_{i=1}^{n} (\omega(k_i)+\kappa) - |P|\Big|\sum_{i=1}^{n} k_i\Big|  \bigg).
	\end{equation}
Since the absolute value is subadditive and $\omega(k)\ge \sfc \abs k$,
	  this is larger than $\kappa$ for $\abs P\le \sfc$, which proves the statement in the case $\Lambda<\infty$.
	The case $\Lambda=\infty$ directly follows from \cref{thm:ren}, since norm resolvent convergence implies convergence of $\inf \sigma_{\sfe\sfs\sfs}(H_{\kappa,\Lambda}(P))$ and $E_{\kappa,\Lambda}(P)$.
\end{proof}
We now identify a compact set in Fock space containing the (normalized) ground states of $H_{\kappa,\Lambda}(P)$.
To this end, we define
\begin{equation}\label{eq:Gset}
	\sG_{r} \coloneqq \Big\{ \psi\in\FS \colon \norm{a_k\psi }\le \frac{r}{\sqrt {\abs k} \vee \abs{k}^2},\ \norm{(a_{k+p}-a_k)\psi} \le
	\frac{r \abs p}{\abs k^2}
	\qquad \mbox{for a.e.}\ k,p\in\IR^3,\ \abs p \le \tfrac{1}{2}\abs k \Big\}.
\end{equation}

\begin{lem}\label{lem:compact}
	For all $r>0$, the set $\sG_{r}$ is pre-compact in $\cF$.
\end{lem}
\begin{proof}
The elements of $\sG_r$ are localized by the first bound, and regular by the second. Conditions of this type are well-known to yield compactness, see \cite[Theorem~3.4]{HaslerHinrichsSiebert.2023} for a detailed proof.
\end{proof}
Now, we prove that $\overline{\sG_{r}}$ contains the ground states of $H_{\kappa,\Lambda}(P)$.
\begin{prop}\label{lem:gscom}
	If
	$\abs P<\sfc$
	, there exist $r>0$ (depending on $\abs P$, $\sfg$ and $\sfc$) such that for all $\kappa>0$ and $\Lambda\in\R_+\cup\{\infty\}$ and any normalized $\psi\in\cD(H_{\kappa,\Lambda}(P))$ with $H_{\kappa,\Lambda}(P)\psi = E_{\kappa,\Lambda}(P)\psi$, we have $\psi\in\overline{\sG_{r}}$.
\end{prop}
\begin{proof}
	Throughout this proof, $r>0$ denotes a (not necessarily fixed) constant solely depending on $\abs P$, $\sfg$ and $\sfc$. Especially, $r$ is independent of $\kappa$ or $\Lambda$.
	Further, fix $\kappa>0$, $\Lambda<\infty$ and $\psi$ as in the statement.

	The starting point of our proof is the the pull-through formula
	\begin{equation}\label{eq:pullthrough}
		a_k \psi = -v_\Lambda(k)R_{\kappa,\Lambda}(P,k)\psi \qquad \mbox{with}\quad R_{\kappa,\Lambda}(P,k)\coloneqq \left(H_{\kappa,\Lambda}(P-k) - E_{\kappa,\Lambda}(P)+\omega (k) + \kappa \right)^{-1},
	\end{equation}
	which holds true for almost every $k\in\IR^3$.
	To check this,
	compute using the commutation relations
	\begin{equation}
	 0=a_k(H_{\kappa, \Lambda}(P)- E_{\kappa,\Lambda}(P))\psi = v_\Lambda(k)\psi + (H_{\kappa, \Lambda}(P-k)+\omega(k) + \kappa -  E_{\kappa,\Lambda}(P))a_k\psi.
	\end{equation}
The formula then follows by applying $R_{\kappa,\Lambda}(P,k)$, which is well defined since $E_{\kappa,\Lambda}(P-k)\ge E_{\kappa,\Lambda}(P) -\abs k \abs P  \ge E_{\kappa,\Lambda}(P)  - \omega(k)$, by \cref{lem:convex} and the assumption $\abs P\le\sfc$, see e.g. \cite{Gerard.2000,Dam.2018} for more details.

To estimate the resolvent, we use Lemma~\ref{lem:convex} to obtain the bounds
\begin{equation}
 E_{\kappa,\Lambda}(P-k)-E_{\kappa,\Lambda}(P) + \omega(k) \ge \omega(k)-\abs P\abs k \ge
 \begin{cases}
       (\sfc-|P|)|k| &  \mbox{for all}\ k\in\IR^3,\\
			\frac{\xi}{2}|k|^2 & \mbox{if}\ \abs k>2 \xi^{-1}\abs P.                                                                                                                                                                          \end{cases}
\end{equation}
Then, using that $H_{\kappa,\Lambda}(P-k)\ge E_{\kappa, \Lambda}(P-k)$, it follows directly from the spectral theorem that
\begin{equation}\label{eq:resbound}
	\norm{R_{\kappa,\Lambda}(P,k)}\le \begin{cases} ((\sfc-\abs P)\abs k)^{-1} & \mbox{for all}\ k\in\IR^3,\\
	2 \xi^{-1} \abs k^{-2} & \mbox{if} \ \abs k>2 \xi^{-1}\abs P.  \end{cases}
\end{equation}
Further employing
 $|v(k)|\le \sfg (\sqrt{\abs{k}/\sfc}\wedge 1)$
we find, for an appropriate choice of $r$,
\begin{equation}\label{eq:vresbound}
	\|v_\Lambda(k)R_{\kappa,\Lambda}(P,k)\| \le \frac{r}{\sqrt{\abs k}\vee {\abs k^2}},
\end{equation}
which combined with \cref{eq:pullthrough} proves the desired bound on $a_k \psi$ in the definition \cref{eq:Gset}.

	Now let $\abs p \le \frac 12 \abs k$.
	The resolvent identity gives
	\begin{equation}\label{eq:pullthroughdiff}
		\begin{aligned}
					a_{k+p}\psi - a_k\psi  =  & \left(v_\Lambda(k+p)-v_\Lambda(k)\right)R_{\kappa,\Lambda}(P,k+p)\psi\\
			& + v_\Lambda(k)R_{\kappa,\Lambda}(P,k+p) [p \cdot (P-k-\dG(\hat p))] R_{\kappa,\Lambda}(P,k)\psi.
		\end{aligned}
	\end{equation}
	 Since $|\nabla v_\Lambda(\ell)|\le C\omega^{-1/2}(k)$ for $\tfrac12|k|\le \ell \le \tfrac32 |k|$, cf. \cref{def:v}, using \cref{eq:resbound} yields
	\begin{equation}\label{eq:diffest11}
	 \|\left(v_\Lambda(k+p)-v_\Lambda(k)\right)R_{\kappa,\Lambda}(P,k+p)\psi\|\le  \frac{8C|p|}{(\sfc-\abs P)|k|\sqrt{\omega(k)}} \le  \frac{r\abs p}{\abs k^2}.
	\end{equation}
As $\norm{ \abs{P-k-\dG(\hat p)}  R_{\kappa,\Lambda}(P,k)} \le 1$, \cref{eq:vresbound} also implies
\begin{equation}\label{eq:diffest2}
 \|v_\Lambda(k)R_{\kappa,\Lambda}(P,k+p) [p \cdot (P-k-\dG(\hat p))] R_{\kappa,\Lambda}(P,k)\psi\|\le \frac{r |p|}{|k|^2}.
\end{equation}
Combining \cref{eq:pullthroughdiff,,eq:diffest11,,eq:diffest2} and
	the definition \cref{eq:Gset}, this proves that $\psi\in \sG_r$ in the case $\Lambda<\infty$.
	
	As a consequence of norm-resolvent convergence  and the uniform gap estimate in \cref{cor:massiveexistence}, the spectral projections of $H_{\kappa, \Lambda}(P)$ converge to those of $H_{\kappa, \infty}(P)$ (cf. \cite[Theorem~VIII.23]{ReedSimon.1972}), whence the ground states of $H_{\kappa,\infty}$ are contained in the closure $\overline{\sG_r}$.
\end{proof}
We conclude with the proof of our main result.
\begin{proof}[{\textbf{Proof of \cref{thm:gsex}}}]
	Let $\psi_{\kappa}$ denote any normalized ground state of $H_{\kappa,\infty}(P)$ for $\kappa>0$.
	Since $\cD(H_{\kappa,\infty}(P))\subset \cD(H(P))$, cf. \cref{lem:core}, we find
	\begin{equation}\label{eq:minimizing}
		0 \le \braket{ \psi_{\kappa},(H(P)-E(P)) \psi_{\kappa} } \le  \braket{\psi_{\kappa},(H_{\kappa,\infty}(P)-E(P))\psi_{\kappa}} = E_{\kappa,\infty}(P)-E_{0,\infty}(P) \xrightarrow{\kappa\downarrow0}0,
	\end{equation}
	by \cref{lem:gsconv}.
	Further, since $(\psi_{\kappa})_{\kappa>0}\subset \overline{\sG_r}$ by \cref{lem:gscom}, which is compact by \cref{lem:compact}, there exists a  zero sequence $(\kappa_n)_{n\in\IN}$ such that the  limit $\psi_{\infty} =\lim_{n\to\infty} \psi_{\kappa_n}$ exists.
	The estimate \cref{eq:minimizing} then
		implies that $\psi_\infty\in\cD(H(P)^{1/2})$, whence $\psi_\infty$ is a minimizer of the closed quadratic form of $H(P)$, and thus an eigenvector, which finishes the proof.
\end{proof}


\section*{Acknowledgements}
BH acknowledges funding by the Ministry of Culture and Science of the State of North Rhine-Westphalia within the project `PhoQC'.
JL thanks Christian Hainzl for discussions on the subject and acknowledges financial support through the EUR-EIPHI Graduate School (ANR-17-EURE-0002).


\bibliographystyle{halpha-abbrv}
\bibliography{00lit}

\end{document}